\documentclass[12pt]{article}

\usepackage[largesc,osf]{newtxtext}
\usepackage{cite,enumitem,mathtools}
\usepackage[amsthm,nosymbolsc,vvarbb]{newtxmath} 
\usepackage[cal=euler,frak=euler]{mathalpha} 
\usepackage{bm} 

\usepackage[margin=2.5cm]{geometry}

\title{Algebras of distributions suitable for \\
       phase-space quantum mechanics. I}

\author{José M. Gracia-Bondía and Joseph C. Várilly \\[12pt]
{\small
Escuela de Matemática, Universidad de Costa Rica,
11501 San José, Costa Rica}}

\date{\small
J. Math.\ Phys. \textbf{29} (1988), 869--879}

\newcommand{\al}{\alpha}            
\newcommand{\bt}{\beta}             
\newcommand{\Dl}{\Delta}            
\newcommand{\dl}{\delta}            
\newcommand{\eps}{\varepsilon}      
\newcommand{\ga}{\gamma}            

\newcommand{\bC}{\mathbb{C}}
\newcommand{\bN}{\mathbb{N}}
\newcommand{\bR}{\mathbb{R}}
\newcommand{\one}{\mathbb{1}}         

\newcommand{\sD}{\mathcal{D}}
\newcommand{\sE}{\mathcal{E}}       
\newcommand{\sF}{\mathcal{F}}       
\newcommand{\sG}{\mathcal{G}}
\newcommand{\sM}{\mathcal{M}}       
\newcommand{\sO}{\mathcal{O}}
\newcommand{\sR}{\mathcal{R}}
\newcommand{\sS}{\mathcal{S}}       
\newcommand{\sV}{\mathcal{V}}

\bmdefine{\sss}{s}                  

\newcommand{\del}{\partial}         
\renewcommand{\geq}{\geqslant}      
\renewcommand{\leq}{\leqslant}      
\newcommand{\ovl}{\overline}        
\newcommand{\ox}{\otimes}           
\newcommand{\stroke}{\mathbin|}
\newcommand{\x}{\times}             
\renewcommand{\:}{\colon}           
\newcommand{\4}{\diamond}           

\newcommand{\abar}{\bar a}
\newcommand{\delhat}{\hat\partial}
\newcommand{\Fbar}{\widetilde{F}}
\newcommand{\sOexp}{\sO_\mathrm{exp}}
\newcommand{\xyx}{\times\cdots\times}

\newcommand{\half}{{\mathchoice{\thalf}{\thalf}{\shalf}{\shalf}}}
\newcommand{\shalf}{{\scriptstyle\frac{1}{2}}} 
\newcommand{\thalf}{\tfrac{1}{2}}   

\newcommand{\hideqed}{\renewcommand{\qed}{}} 

\newcommand{\set}[1]{\{\,#1\,\}}     
\newcommand{\word}[1]{\quad\text{#1}\quad} 

\newcommand{\duo}[2]{\langle#1,#2\rangle} 

\newcommand{\pd}[2]{\frac{\partial#1}{\partial#2}} 
\newcommand{\scal}[2]{(#1\stroke#2)} 

\newcommand{\twobytwo}[4]{\begin{pmatrix} 
         #1 & #2 \\ #3 & #4 \end{pmatrix}}

\theoremstyle{plain}
\newtheorem{thm}{Theorem}             
\newtheorem{lema}{Lemma}              
\newtheorem{corl}{Corollary}          
\newtheorem{prop}{Proposition}        

\theoremstyle{definition}
\newtheorem{defn}{Definition}         

\theoremstyle{remark}
\newtheorem{remk}{Remark}             

\makeatletter
\renewcommand{\section}{\@startsection{section}{1}{\z@}%
							 {-3.25ex \@plus -1ex \@minus -.2ex}%
							 {1.5ex \@plus.2ex}%
							 {\normalfont\large\bfseries}}
\renewcommand{\subsection}{\@startsection{subsection}{2}{\z@}%
							 {-3.25ex \@plus -1ex \@minus -.2ex}%
							 {1.5ex \@plus .2ex}%
							 {\normalfont\normalsize\bfseries}}
\makeatother

\begin{document}

\maketitle

\begin{abstract}
The twisted product of functions on $\bR^{2N}$ is extended to a
$*$-algebra of tempered distributions which contains the rapidly
decreasing smooth functions, the distributions of compact support, and
all polynomials, and moreover is invariant under the Fourier
transformation. The regularity properties of the twisted product are
investigated. A matrix presentation of the twisted product is given,
with respect to an appropriate orthonormal basis, which is used to
construct a family of Banach algebras under this product.
\end{abstract}

\section{Introduction} 
\label{sec:intro}

This is the first of two papers whose aim is to give a rigorous
formulation to the Weyl--Wigner--Groenewold--Moyal or phase-space
approach to quantum mechanics of spinless, nonrelativistic particles.
(In a future article, we will show how spin may be incorporated also
in this formalism.) In recent years, this approach has received
increasing attention \cite{GrossmannLS68, BayenFFLS78, AmietH81,
HilleryOCSW84, HAtom}. However, much remains to be done to unify its
different strands. On the one hand, much useful quantum physics can be
done using the distribution functions in the sense of
Wigner~\cite{HilleryOCSW84}. On the other hand, most mathematical
attention has centred on the Weyl operator calculus \cite{Pool66,
Anderson72, Cressman76, Hormander79}. As Groenewold
\cite{Groenewold46} and Moyal~\cite{Moyal49} have shown, one can work
with functions on the classical phase space only, in a self-contained
way, using the ``twisted product'' concept. Similarly, the
practitioners of ``deformation theory'' \cite{BayenFFLS78} have given
a promising axiomatic basis for quantum mechanics, but with
mathematical tools rather different from the usual functional-analytic
methods of quantum theory.

We attempt here to establish a mathematically rigorous and physically
manageable formulation for quantum mechanics in phase space. To obtain
the right mathematical context, we must, for example, specify those
pairs of functions whose twisted product may be formed; and a suitable
function space should include as many observables of physical interest
as possible. To include the basic observables of position and
momentum, we must leave aside the algebra of bounded operators on a
Hilbert space: in rigorizing the phase-space approach, one soon finds
that it is useful to work with locally convex topological vector
spaces which are not necessarily Banach spaces.

The paper is organized as follows. In
Sec.~\ref{sec:Schwartz-algebras}, we review the properties of the
twisted product and convolution in the Schwartz space $\sS(\bR^{2N})$.
In Sec.~\ref{sec:Moyal-algebra}, we dualize these notions to the case
where one or both factors are tempered distributions, and identify the
Moyal algebra $\sM$, that is, the largest $*$-algebra of distributions
where these operations are defined and associative. We show that $\sM$
is invariant under Fourier transformations. In
Sec.~\ref{sec:Moyal-regularity}, we consider the regularity properties
of the twisted product and convolution, and show that distributions of
compact support belong to the Moyal algebra. In
Sec.~\ref{sec:matrix-basis}, we construct an orthonormal basis in
$\sS(\bR^{2N})$; using this basis, we show that the twisted product
may be presented as a matrix product of double sequences. As a
consequence, we construct a net of Sobolev-like spaces of tempered
distributions, some of which are Banach algebras with respect to the
twisted product; these permit a more detailed examination of the Moyal
algebra~$\sM$.

\section{The algebras ($\sS_2,\x$) and ($\sS_2,\4$)} 
\label{sec:Schwartz-algebras}

Throughout this paper, we work with certain spaces of functions and
distributions over $\bR^{2N}$, regarded as the phase space
$T^*(\bR^N)$. For $u,v \in \bR^{2N}$, we write $u'v$ and $u'Jv$ for
the ordinary and symplectic scalar products of~$u$ and~$v$. Choosing
and fixing an orthonormal symplectic basis for $\bR^{2N}$, we write
$$
u = (u_1,u_2,\dots,u_{2N}) = (q_1,\dots,q_N,p_1,\dots,p_N)
$$
and
$$
v = (v_1,v_2,\dots,v_{2N}) 
= (\tilde q_1,\dots,\tilde q_N,\tilde p_1,\dots,\tilde p_N),
$$
where explicitly
$$
u'v := \sum_{i=1}^N (q_i\tilde q_i + p_i\tilde p_i), \qquad
u'Jv := \sum_{i=1}^N (q_i\tilde p_i - p_i\tilde q_i),
$$
where $J$ is the matrix $\twobytwo{0}{1_N}{-1_N}{0}$ in the chosen
basis. Note that $v'Ju = -u'Jv$ and $u'Ju = 0$.

We define $\sS_2 := \sS(\bR^{2N})$ as the Schwartz space of smooth
rapidly decreasing functions on $\bR^{2N}$. If $f \in \sS_2$,
$s \in \bR^{2N}$ and $1 \leq j \leq 2N$, we define
$f^*(u) := \ovl{f(u)}$, $\check f(u) := f(-u)$, and also:
\begin{alignat*}{2}
(\mu_j f)(u) &:= u_j f(u), &\qquad     \del_j f &:= \pd{f}{u_j},
\\
(\tau_sf)(u) &:= f(u - s), &\qquad (\eps_sf)(u) &:= e^{is'Ju}f(u), 
\end{alignat*}
and
$$
\delhat_j f := \begin{cases}
 \del_{j+N} f & \text{if $1 \leq j \leq N$,} \\
-\del_{j-N} f & \text{if $N < j \leq 2N$.}   \end{cases}
$$

We make three normalizations which are a little unconventional. First,
for integrals over $\bR^{2N}$ we use the Haar measure 
$dx := (2\pi)^{-N}\,d^{2N}x$ where $d^{2N}x$ is Lebesgue measure.
(This gets rid of powers of~$2\pi$ in Fourier
transforms~\cite{Rudin73}). In particular, $\int e^{-x^2/2} \,dx = 1$.
Secondly, we use the \textit{bilinear} form
$$
\duo{f}{g} := \int f(x) g(x) \,dx
$$
and the \textit{sesquilinear} form
$$
\scal{f}{g} := 2^{-N} \duo{f^*}{g} = 2^{-N} \int \ovl{f(x)} g(x) \,dx
$$
whenever the integrals converge. For $f \in L^2(\bR^{2N})$, we will
use the norm $\|f\| := \scal{f}{f}^{1/2}$. Thirdly, for Planck's
constant we take $\hbar = 2$ (rather than the usual $\hbar = 1$).

We define an ordinary Fourier transform $\sF$ and two symplectic
Fourier transforms $F$ and~$\Fbar$ \cite{MiracleS66} by
$$
(\sF f)(u) := \int f(t) e^{-it'u} \,dt,  \quad
(Ff)(u) := \int f(t) e^{-it'Ju} \,dt,  \quad
(\Fbar f)(u) := \int f(t) e^{it'Ju} \,dt.
$$
The transforms $\sF$, $F$ and $\Fbar$ are commuting isomorphisms (of
Fréchet spaces) of $\sS_2$ onto $\sS_2$, and satisfy the following
formulas:
\begin{alignat*}{2}
Ff(u) &= \sF f(Ju), &\quad \Fbar f(u) &= \sF f(-Ju),
\\
F^2 &= \Fbar^2 = \text{Id}, 
&\quad \Fbar f &= (Ff)^\vee = F(\check f), \quad (Ff)^* = \Fbar(f^*),
\\
F(\tau_s f) &= \eps_{-s} Ff, &\quad F(\eps_s f) &= \tau_{-s} Ff,
\\
F(\delhat_j f) &= -i\mu_j Ff, &\quad F(\mu_j f) &= i\delhat_j Ff,
\\
\duo{Ff,g} &= \duo{f,\Fbar g}, &\quad \scal{Ff}{g} &= \scal{f}{Fg}.
\end{alignat*}

\begin{defn} 
\label{df:twisted-product}
If $f,g \in \sS_2$, the \textit{twisted product} $f \x g$ is defined
by
\begin{align}
(f \x g)(u)
&:= \iint f(v) g(w) \exp\bigl( i(u'Jv + v'Jw + w'Ju) \bigr) \,dv \,dw
\nonumber \\
&= \iint f(u + s) g(u + t) \,e^{is'Jt} \,ds \,dt. 
\label{eq:twisted-product} 
\end{align}
The \textit{twisted convolution} $f \4 g$ is defined by
\begin{equation}
(f \4 g)(u) := \int f(u - t) g(t) \,e^{-iu'Jt} \,dt.
\label{eq:twisted-colvoln} 
\end{equation}
\end{defn}

\begin{remk} 
It was von Neumann \cite{Neumann31} who introduced the twisted
convolution (although he gave it no name) in order to establish the
uniqueness of the Schroedinger representation. It has been used by
Kastler and others~\cite{MiracleS66,Kastler65,LoupiasMS66} to study
the canonical commutation relations.
\end{remk}

\begin{remk} 
The twisted product is nothing but the Weyl functional
calculus~\cite{Anderson72} seen from another point of view. As
in~\cite{GrossmannLS68} and elsewhere, one may regard it as
$f \x g = \sF^{-1}(\sF f \4 \sF g)$; but perhaps a more natural
motivation is the following. The pointwise product $f(u)g(u)$ is not
suitable for quantum mechanics since the uncertainty principle forbids
localization at a point in phase space. Following
S{\l}awianowski~\cite{Slawianowski77}, we seek to replace it by some
other product which is translation and symplectic equivariant,
associative, and nonlocal. In~\cite{AmietH81} it is shown that the
only integral kernels satisfying translation and symplectic
equivariance and associativity are $a\dl(s)\dl(t)$ -- for the
pointwise product -- and $be^{ics'Jt}$, where $a,b,c$ are constants
which we may set equal to~$1$.
\end{remk}

\begin{prop} 
\label{pr:Leibniz-rules}
If $f,g \in \sS_2$, then $f \x g \in \sS_2$, the map 
$(f,g) \mapsto f \x g$ is a continuous bilinear operation on $\sS_2$,
and
\begin{align}
\del_j(f \x g) &= \del_j f \x g + f \x \del_j g;
\nonumber \\ 
\mu_j(f \x g) &= f \x \mu_j g + i\delhat_j f \x g 
= \mu_j f \x g - if \x \delhat_j g.
\label{eq:Leibniz-rules} 
\end{align}
\end{prop}

\begin{proof}
The Leibniz formula follows by differentiating
\eqref{eq:twisted-product} under the integral sign, and
\eqref{eq:Leibniz-rules} is a straightforward calculation. By
induction on these formulas, $f \x g$ lies in $\sS_2$. If
$\al = (\al_1,\dots,\al_{2N}) \in \bN^{2N}$, we write
$\del^\al = \del_1^{\al_1} \cdots \del_{2N}^{\al_{2N}}$ and similarly
define $\mu^\al$, $\delhat^\al$. Then
$$
\mu^\al\,\del^\ga(f \x g) = \sum_{\bt\leq\al} \sum_{\eps\leq\ga}
(-i)^{|\bt|} \binom{\al}{\bt} \binom{\ga}{\eps} 
\mu^{\al-\bt} \,\del^{\ga-\eps} f \x \delhat^\bt \,\del^\eps g.
$$
{}From \eqref{eq:twisted-product} we get
$\|f \x g\|_\infty \leq \|f\|_1 \,\|g\|_1$. Since the topology
of~$\sS_2$ is given by the seminorms
$p_{\al\ga}(f) := \|\mu^\al\,\del^\ga f\|_\infty$ or by
$q_{\al\ga}(f) := \|\mu^\al\,\del^\ga f\|_1$, the estimates
$$
p_{\al\ga}(f \x g) \leq \sum_{\bt\leq\al} \sum_{\eps\leq\ga} 
\binom{\al}{\bt} \binom{\ga}{\eps} q_{\al-\bt,\ga-\eps}(f)
q_{0,\eta+\eps}(g),
$$
with $\eta_j = \bt_{j\pm N}$ for all~$j$, show that
$(f,g) \mapsto f \x g$ is jointly continuous for the topology 
of~$\sS_2$.
\end{proof}

The various Fourier transforms intertwine $\x$~and~$\4$, just as with
``ordinary'' products and convolutions. In fact, even more is true: by
applying a symplectic Fourier transform to one side only, we can
interchange the operations $\x$~and~$\4$. This allows us to work with
the operation most convenient to any particular calculation,
transferring the result to the other one by Fourier-invariance
of~$\sS_2$. Explicitly, we find:
\begin{equation}
f \x g = Ff \4 g = f \4 \Fbar g, \qquad
f \4 g = Ff \x g = f \x \Fbar g,
\label{eq:Fourier-switch} 
\end{equation}
since, for example,
\begin{align*}
(f \x g)(u) &= \iint f(v) g(w) \,e^{-iv'J(u-w)} e^{iw'Ju} \,dv \,dw
\\
&= \int Ff(u - w) g(w) \,e^{-iu'Jw} \,dw = (Ff \4 g)(u). 
\end{align*}
We also find 
\begin{equation}
\sF(f \x g) = \sF f \4 \sF g,  \qquad  \sF(f \4 g) = \sF f \x \sF g,
\label{eq:Fourier-transfer} 
\end{equation}
and exactly analogous formulas with $\sF$ replaced by $F$ or~$\Fbar$.
Also,
\begin{equation}
(f \x g) \x h = f \x (g \x h), \qquad (f \4 g) \4 h = f \4 (g \4 h),
\label{eq:twisted-assoc} 
\end{equation}
since 
\begin{align*}
((f \4 g) \4 h)(u) 
&= \iint f(u - t - s) g(s) h(t) \,e^{-i(u'Jt+(u-t)'Js)} \,ds \,dt
\\
&= \iint f(u - v) g(v - t) h(t) \,e^{-i(u'Jv-t'Jv)} \,dt \,dv
\\
&= (f \4 (g \4 h))(u)
\end{align*}
and applying \eqref{eq:Fourier-transfer} yields the associativity
of~$\x$. Next,
$$
(f \x g)^* = g^* \x f^*, \qquad (f \4 g)^* = g^* \4 f^*,
$$
since, for instance, 
\begin{align*}
(f \x g)^*(u) 
&= \iint f^*(u + s) g^*(u + t) \,e^{-is'Jt} \,ds \,dt
\\
&= \iint g^*(u + t) f^*(u + s) \,e^{it'Js} \,dt \,ds
\\
&= (g^* \x f^*)(u).
\end{align*}

A fact of fundamental importance is the following identity.

\begin{prop} 
\label{pr:tracial-property}
If $f,g \in \sS_2$, then
\begin{equation}
\int (f \x g)(u) \,du = \int (g \x f)(u) \,du = \int f(u) g(u) \,du.
\label{eq:tracial-property} 
\end{equation}
\end{prop}

\begin{proof} 
\begin{align*}
\int (f \x g)(u) \,du
&= \sF(f \x g)(0) = (\sF f \4 \sF g)(0) = \int \sF f(-t) \sF g(t) \,dt
\\
&= (\sF f * \sF g)(0) = \sF(fg)(0) = \int f(u) g(u) \,du
\end{align*} 
where $*$ denotes ordinary convolution.
\end{proof}

The cyclicity inherent in the tracial
identity~\eqref{eq:tracial-property} is what allows us to push through
the extension via duality. We note an important consequence
of~\eqref{eq:tracial-property}.

\begin{prop} 
\label{pr:product-transfer}
If $f,g,h \in \sS_2$, then
\begin{align}
\duo{f \x g}{h} &= \duo{f}{g \x h} = \duo{g}{h \x f};
\label{eq:product-transfer} 
\\
\duo{f \4 g}{h} &= \duo{f}{\check g \4 h} = \duo{g}{h \4 \check f};
\label{eq:convol-transfer} 
\\
\scal{h}{f \x g} &= \scal{f^* \x h}{g} = \scal{h \x g^*}{f}.
\nonumber
\end{align}
\end{prop}

\begin{proof}
{}From \eqref{eq:tracial-property} and the 
associativity~\eqref{eq:twisted-assoc}, we find that all three 
expressions in~\eqref{eq:product-transfer} are equal to
$\int (f \x g \x h)(u) \,du$. Now \eqref{eq:convol-transfer} follows
from~\eqref{eq:Fourier-switch}, and the third formula is immediate.
\end{proof}

\section{Duality and the Moyal algebra} 
\label{sec:Moyal-algebra}

Having established a calculus for functions in $\sS_2$ with twisted
product and convolution, we now extend it to a larger algebra of
tempered distributions. First we consider the twisted product or
convolution of a tempered distribution and a test function.

For $T \in \sS_2'$, $h \in \sS_2$, we write $\duo{T}{h} := T(h)$. We
also extend our previous notations in the usual way:
\begin{align*}
\duo{\sF T}{h} &:= \duo{T}{\sF h}, \quad
\duo{\check T}{h} := \duo{T}{\check h},
\\
\scal{T}{h} &:= 2^{-N}\duo{T^*}{h} := 2^{-N}\duo{T}{h^*}^*,
\\
\duo{FT}{h} &:= \duo{T}{\Fbar h}, \quad
\duo{\Fbar T}{h} := \duo{T}{Fh},
\\
\duo{\del_j T}{h} &:= -\duo{T,\del_j h}, \quad
\duo{\mu_j T}{h} := \duo{T}{\mu_j h}.
\end{align*}

\begin{defn} 
\label{df:transposed-product}
For $T \in \sS_2'$, $f,h \in \sS_2$, we define $T \x f$, $f \x T$,
$T \4 f$ and $f \4 T$ in $\sS_2'$ by
\begin{alignat}{2}
\duo{T \x f}{h} &:= \duo{T}{f \x h}, 
&\quad \duo{f \x T}{h} &:= \duo{T}{h \x f},
\label{eq:transposed-product} 
\\
\duo{T \4 f}{h} &:= \duo{T,\check f \4 h}, 
&\quad \duo{f \4 T}{h} &:= \duo{T}{h \4 \check f}.
\nonumber
\end{alignat}
The continuity of $\x$ and $\4$ in $\sS_2$ implies that each right
hand side is continuous and linear in $h$ and thus defines an element
of~$\sS_2'$. By \eqref{eq:product-transfer}
and~\eqref{eq:convol-transfer}, these are extensions of the
corresponding operations on~$\sS_2$.
\end{defn}

Throughout this paper, every dual space $E'$ of a locally convex space
$E$ is topologized by the strong dual topology, that of uniform
convergence on bounded subsets of~$E$. Then the four bilinear maps
$\: \sS_2' \x \sS_2 \to \sS_2'$ defined above are hypocontinuous.
Indeed, since $\sS_2$ and $\sS_2'$ are barrelled it
suffices~\cite{Horvath66,Schaefer66} to check separate continuity. For
example, for fixed~$T$, the map $f \mapsto \duo{T}{f \x h}$ is
continuous, uniformly so for $h$ in a bounded subset of~$\sS_2$ (by
the joint continuity of~$f$ and~$h$), so that $f \mapsto T \x f$ is
continuous from $\sS_2$ to~$\sS_2'$. For fixed $f$, $T \mapsto T \x f$
is the transpose of the continuous map $h \mapsto f \x h$ of $\sS_2$
into $\sS_2$, and as such is continuous from $\sS_2'$ to~$\sS_2'$.

All formulas of Sec.~\ref{sec:Schwartz-algebras} involving $f$ and $g$
extend to analogous formulae for $T$ and $f$ (e.g.,
$T \x f = FT \4 f = T \4 \Fbar f$). This is easily checked since $A =
B$ in $\sS_2'$ iff $\duo{A}{h} = \duo{B}{h}$ for all $h \in \sS_2$,
and we may reduce to the $\sS_2$ case
using~\eqref{eq:transposed-product}.

We write $\one$ for the constant function with value~$1$, and $\dl$
for the Dirac measure of mass one supported at~$0$. These are the
identities for the operations $\x$ and~$\4$:
$$
\one \x f = f \x \one = f,  \qquad  \dl \4 f = f \4 \dl = f.
$$
This follows from
$$
\duo{\one \x f}{h} = \duo{\one}{f \x h} = \int (f \x h)(u) \,du
= \duo{f}{h}
$$
by \eqref{eq:tracial-property}: thus $\one \x f = f$ as elements of
$\sS_2'$. We show below that $\one \x f$ is continuous, so 
$\one \x f = f$ as functions in $\sS_2$. The other half of the
equation follows from~\eqref{eq:Fourier-transfer}, since
$F\dl = \Fbar\dl = \sF\dl = \one$. From~\eqref{eq:Fourier-switch} we
also obtain the formulas
\begin{equation}
\one \4 f = \dl \x f = \Fbar f, \qquad f \4 \one = f \x \dl = Ff.
\label{eq:Fourier-cross} 
\end{equation}
Using the fact that $u_j = \mu_j\one$, we obtain
from~\eqref{eq:Leibniz-rules} the important identities
$$
u_j \x f = \mu_j f + i\,\delhat_j f,  \qquad 
f \x u_j = \mu_j f - i\,\delhat_j f,
$$
which in the $(q,p)$ notation become
\begin{alignat}{2}
q_j \x f &= \biggl( q_j + i\pd{}{p_j} \biggr) f, 
&\qquad p_j \x f &= \biggl( p_j - i\pd{}{q_j} \biggr) f,
\nonumber \\[\jot]
f \x q_j &= \biggl( q_j - i\pd{}{p_j} \biggr) f, 
&\qquad f \x p_j &= \biggl( p_j + i\pd{}{q_j} \biggr) f.
\label{eq:Weyl-operators} 
\end{alignat}

If $T \in \sS_2'$ and $f \in \sS_2$, the ordinary convolution $T*f$
is~\cite{Horvath66,Schwartz66} a smooth function in~$\sO_C$, whereas
the ordinary product $Tf$ is a ``rapidly decreasing distribution''
in~$\sO_C'$ but need not be smooth. In contrast,
\eqref{eq:Fourier-transfer} -- extended to~$\sS_2'$ -- shows that the
twisted product and twisted convolution have similar properties of
smoothness and of growth at infinity. To see this, we first note that
\eqref{eq:twisted-colvoln} may be rewritten as
$$
(f \4 g)(u) = \duo{\eps_{-u} \tau_u \check f}{g} 
= \duo{f}{\eps_u \tau_u \check g}.
$$
Thus, in convolution formulas such as
$(T*f)(u) = \duo{T}{\tau_u \check f}$, the translations $\tau_u$ are
replaced by $\eps_u \tau_u$ or $\eps_{-u} \tau_u$.

\begin{thm} 
\label{th:product-formulas}
If $T \in \sS_2'$, $f \in \sS_2$, then $T \x f$, $f \x T$, $T \4 f$
and~$f \4 T$ are smooth functions on~$\bR^{2N}$, given by
\begin{align}
(T \x f)(u) &= \duo{T}{\eps_u \tau_u Ff}, 
\nonumber \\
(f \x T)(u) &= \duo{T}{\eps_{-u} \tau_u\Fbar f} 
\label{eq:product-formula} 
\\
(T \4 f)(u) &= \duo{T}{\eps_u \tau_u\check f}, 
\nonumber \\
(f \4 T)(u) &= \duo{T}{\eps_{-u} \tau_u\check f}.
\label{eq:convol-formula} 
\end{align}
\end{thm}

\begin{proof}
If $h \in \sS_2$, the maps $u \mapsto \tau_u h$, $u \mapsto \eps_u h$
are continuous from $\bR^{2N}$ to~$\sS_2$, so the right hand sides of
these formulas are jointly continuous in $u$ and~$f$. By
transposition, since
$\duo{T}{\eps_u \tau_u\check f} = \duo{\eps_{-u} \tau_u \check T}{f}$,
they are also continuous in~$T$. These right hand sides define
separately continuous extensions to $\sS_2' \x \sS_2$ of the twisted
product and convolution on the dense subspace $\sS_2 \x \sS_2$; since
these extensions are necessarily unique, they coincide with $T \x f$,
etc., as defined earlier.

Now
$$
\del_j h(u) 
= \pm \lim_{c\to 0} c^{-1} \bigl( h(u + ce_j) - h(u) \bigr),
$$
where $e_k$ is the $k$th basis vector in $\bR^{2N}$, whenever this
limit exists. By calculation, we find that
$$
\lim_{c\to 0} c^{-1}
\duo{T}{\eps_{u+ce_j} \tau_{u+ce_j} Ff - \eps_u \tau_u Ff} 
= (\del_j T \x f)(u) + (T \x \del_j f)(u)
$$
as expected, so by induction $T \x f$ is infinitely differentiable, 
with the Leibniz formula 
$\del_j(T \x f) = \del_j T \x f + T \x \del_j f$ holding as an 
equality between smooth functions. The other three cases are similar.
\end{proof}

Let $\sE_2$ denote the space of smooth functions on $\bR^{2N}$, with
the topology of uniform convergence of all derivatives on compact
sets. Ordinary convolution operators are precisely those which commute
with translations; we may characterize the twisted convolution
operators as those which commute with ``twisted translations''.

\begin{thm} 
\label{th:twisted-translations}
Let $L \: \sS_2 \to \sE_2$ be linear and continuous; then there is a
unique $T \in \sS_2'$ with $L(f) = T \x f$ for all $f \in \sS_2$ iff
$L$ commutes with $\set{\eps_u \tau_u : u \in \bR^{2N}}$.
\end{thm}

\begin{proof}
From \eqref{eq:product-formula} we find that 
\begin{align*}
\eps_v \tau_v(T \x f)(u)
&= e^{iv'Ju} (T \x f)(u - v) 
= e^{iv'Ju} \duo{T}{\eps_{u-v} \tau_{u-v} Ff}
\\
&= e^{iv'Ju} \duo{T}{t \mapsto e^{i(u-v)'Jt} Ff(t - u + v)}
\\
&= \duo{T}{t \mapsto e^{i(u'Jt - v'J(t-u))} \eps_{-v} Ff(t - u)}
\\
&= \duo{T}{t \mapsto e^{iu'Jt} \tau_{-v} \eps_{-v} Ff(t - u)} 
 = \duo{T}{\eps_u \tau_u F(\eps_v \tau_v f)}
\end{align*}
so that $f \mapsto T \x f$ commutes with any~$\eps_v \tau_v$.

On the other hand, given $L \: \sS_2 \to \sE_2$ which commutes with
all $\eps_v \tau_v$, we define $T \in \sS_2'$ by
$\duo{T}{h} := L(Fh)(0)$. (So $T$ is unique.) For $u \in \bR^{2N}$, 
$f \in \sS_2$, we then obtain
\begin{align*}
(T \x f)(u) &= \duo{T,\eps_u \tau_u Ff} 
= \duo{T}{F(\eps_{-u} \tau_{-u}f)} 
\\
&= L(\eps_{-u} \tau_{-u}f)(0) = \eps_{-u} \tau_{-u} (Lf)(0)
\\
&= \tau_{-u}(Lf)(0) = (Lf)(u).
\tag*{\qed}
\end{align*}
\hideqed
\end{proof}

We can now define the Moyal $*$-algebra $\sM$. We define it as the
intersection of two spaces $\sM_L$ and $\sM_R$ which, to the best of
our knowledge, were first considered by Antonets~\cite{Antonets78}.

\begin{defn} 
\label{df:Moyal-algebra}
\begin{enumerate}
\item[(1)]
$\sM_L := 
\set{S \in \sS_2' : S \x f \in \sS_2 \text{ for all } f \in \sS_2}$;
\item[(2)]
$\sM_R := 
\set{R \in \sS_2' : f \x R \in \sS_2 \text{ for all } f \in \sS_2}$;
\item[(3)]
$\sM := \sM_L \cap \sM_R$.
\end{enumerate}
\end{defn}

Note that $S \in \sM_L$ iff $S^* \in \sM_R$ since
$(S \x f)^* = f^* \x S^*$. Since $\sS_2$ is a Fréchet space, the
maps $f \mapsto S \x f$, $f \mapsto f \x R$ are continuous from
$\sS_2$ to $\sS_2$ by the closed graph theorem.

It is clear that $\sS_2 \subset \sM$ and that $\one \in \sM$. The
formulas~\eqref{eq:Fourier-cross} show that $\dl \in \sM$. Now, if
$S \in \sM_L$ and $f \in \sS_2$, we have
\begin{align*}
\del_j S \x f &= \del_j(S \x f) - S \x \del_j f,
\\
\mu_j S \x f &= \mu_j(S \x f) + iS \x \delhat_j f,
\end{align*}
and so $\del_j S \in \sM_L$ and $\mu_j S \in \sM_L$; thus $\sM_L$, and
similarly $\sM_R$ and $\sM$, is closed under partial differentation
and multiplication by polynomials. Hence, in particular, all
polynomials lie in~$\sM$.

We now extend the twisted product to the case of one distribution
in~$\sM$ and one in~$\sS_2'$ (so that $\sS_2'$ is an $\sM$-bimodule).

\begin{defn} 
\label{df:Moyal-bimodule}
If $R \in \sM_R$, $S \in \sM_L$, $T \in \sS_2'$, we define $T \x S$, 
$R \x T$ in $\sS_2'$ by
\begin{equation}
\duo{T \x S}{h} := \duo{T}{S \x h}, \qquad
\duo{R \x T}{h} := \duo{T}{h \x R},
\label{eq:Moyal-bimodule} 
\end{equation}
for all $h \in \sS_2$. Since the right hand sides are continuous
in~$h$, $T \x S$ and $R \x T$ are defined in~$\sS_2'$.
\end{defn}

If $R,S \in \sM$, $T \in \sS_2'$ and $f,g,h \in \sS_2$, we may compute:
\begin{align*}
\duo{(T \x f) \x g}{h} &= \duo{T \x f}{g \x h} = \duo{T}{f \x g \x h} 
= \duo{(T \x (f \x g)}{h},
\\
\duo{(R \x S) \x f}{h} &= \duo{R \x S}{f \x h} = \duo{R}{S \x f \x h}
= \duo{R \x (S \x f)}{h}.
\end{align*}
In particular, $(R \x S) \x f \in \sS_2$ for $f \in \sS_2$, so
$R \x S \in \sM_L$. Then
$$
\duo{(T \x R) \x S}{h} = \duo{T \x R}{S \x h} = \duo{T}{R \x S \x h}
= \duo{T \x (R \x S)}{h}.
$$
We conclude that $\sM$ is an associative algebra; in fact, it is a
$*$-algebra since, for $R,S \in \sM$,
\begin{align*}
\duo{(R \x S)^*}{h}
&= \duo{R \x S}{h^*}^* = \duo{R}{S \x h^*}^* 
\\
&= \duo{R^*}{h \x S^*} = \duo{S^* \x R^*}{h}.
\end{align*}•
We may also note that since $S \4 f = S \x \Fbar f$ and
$f \4 R = Ff \x R$, we have 
\begin{align*}
\sM_L 
&= \set{S \in \sS_2' : S \4 f \in \sS_2 \text{ for all } f \in \sS_2},
\\
\sM_R 
&= \set{R \in \sS_2' : f \4 R \in \sS_2 \text{ for all } f \in \sS_2},
\end{align*}
so $\sM$ is also a $*$-algebra under $\4$, where we define
$\duo{T \4 S}{h} := \duo{T}{\check S \4 h}$ and
$\duo{R \4 T}{h} := \duo{T}{h \4 \check R}$. The invariance of~$\sM$
under the several Fourier transforms now follows easily. One easily
checks that the formulas of Sec.~\ref{sec:Schwartz-algebras} remain
valid when $f,g$ are replaced by $R,S \in \sM$.

\begin{remk} 
We show below that $\set{f \x g : f,g \in \sS_2}$ equals $\sS_2$. Thus
$\sM$ is the maximal $*$-algebra which we may define by duality. For,
if $T \in \sS_2'$ with $T \x f, f \x T \in \sM$ for all $f \in \sS_2$,
then by writing $f = g \x h$ we see that $T \x f$ and $f \x T$ both
lie in $\sS_2$, since $T \x f = (T \x g) \x h$ and
$f \x T = g \x (h \x T)$; hence $T \in \sM$.
\end{remk}

\begin{remk} 
$\sM$, $\sM_L$, $\sM_R$ and $\sS_2'$ are distinct spaces of
distributions.
\end{remk}

\section{Regularity properties} 
\label{sec:Moyal-regularity}

In this section we consider in more detail the growth conditions on
resultants of twisted products or convolutions. We identify a space of
smooth functions, $\sO_T$, which contains all functions defined by
\eqref{eq:product-formula} and~\eqref{eq:convol-formula}, and we show
that $\sO_T$ is a normal space of distributions. As a consequence, its
dual space $\sO'_T$ contains all distributions of compact support and
is contained in~$\sM$.

If $(E_i)_{i\in I}$ is a collection of locally convex spaces, the
\textit{projective topology} on the intersection
$E := \bigcap_{i\in I} E_i$ is the weakest locally convex topology
such that all inclusions $E \subset E_i$ are continuous. The
\textit{inductive topology} on the union $F := \bigcup_{i\in I} E_i$
is the strongest locally convex topology such that all inclusions
$E_i \subset F$ are continuous. We will use the projective topology on
decreasing intersections, and the inductive topology on increasing
unions, without further comment.

For $f \in C^m(\bR^{2N})$, $k,m \in \bN$, let
\begin{equation}
p_{k,m}(f) := \sup\set{(1+u^2)^{-k-|\al|/2} |\del^\al f(u)|
: u \in \bR^{2N},\ |\al| \leq m}
\label{eq:Weyl-seminorms} 
\end{equation}
(where $u^2 = u'u = u_1^2 +\cdots+ u_{2N}^2$), and let $\sV_k^m$ be
the space of all $f \in C^m$ such that the function
$(1 + u^2)^{-k-|\al|/2} \,\del^\al f(u)$ vanishes at infinity for all
$|\al| \leq m$, normed by $p_{k,m}$. Now let
\begin{equation}
\sV_k := \bigcap_{m\in\bN} \sV_k^m, \qquad
\sO_T := \bigcup_{k\in\bN} \sV_k.
\label{eq:OT-filtration} 
\end{equation}

Let
$$
q_{k,m}(f) := \sup\set{(1 + u^2)^{-k} |\del^\al f(u)|
: u \in \bR^{2N},\ |\al| \leq m};
$$
then the set $\set{f \in C^m : (1 + u^2)^{-k}\,\del^\al f(u)
\text{ vanishes at infinity for } |\al| \leq m}$, normed by~$q_{k,m}$,
is Horváth's space $\sS_{-k}^m$ \cite{Horvath66}. Using the notations 
of~\cite{Horvath66}, we have
\begin{alignat}{2}
\sS_{-k} &:= \bigcap_{m\in\bN} \sS_{-k}^m,
&\qquad \sO_C &:= \bigcup_{k\in\bN} \sS_{-k},
\nonumber \\
\sO_C^m &:= \bigcup_{k\in\bN} \sS_{-k}^m,
&\qquad \sO_M &:= \bigcap_{m\in\bN} \sO_C^m.		
\label{eq:OC-filtration} 
\end{alignat}
Here $\sO_M$ consists of smooth functions for which each derivative is
polynomially bounded; $\sO_C$ is a subset of $\sO_M$, wherein the
degree of the polynomial bound is independent of the derivative; and
$\sO_T$ is the subset of $\sO_M$ wherein that degree increases
linearly with the order of the derivative. Indeed, this leads to the 
following proposition.

\begin{prop} 
\label{pr:OT-sequence}
$\sO_C \subset \sO_T \subset \sO_M$ with continuous inclusions.
\end{prop}

The proof is easy and will be omitted.

\begin{thm} 
\label{th:Moyal-regularity}
If $T \in \sS_2'$, $f \in \sS_2$, then $T \x f$, $f \x T$, $T \4 f$
and $f \4 T$ all lie in~$\sO_T$. Moreover, these four bilinear maps of
$\sS_2' \x \sS_2$ into~$\sO_T$ are separately continuous.
\end{thm}

\begin{proof}
It suffices to consider the case of $T \4 f$.

Differentiating the equality
$\eps_u \tau_u(T \4 f) = T \4 (\eps_u \tau_u f)$, with $u = tJe_j$, at
$t = 0$, we get 
$$
(\delhat_j + i\mu_j)(T \4 f) = T \4 (\delhat_j + i\mu_j)f,
$$
so that 
$$
\delhat_j(T \4 f) = T \4 (\delhat_j + i\mu_j)f - i\mu_j(T \4 f),
$$
and by induction we get, for $\al \in \bN^{2N}$,
\begin{equation}
\delhat^\al(T \4 f)(u) = \sum_{\bt\leq\al} P_\bt(u) (T \4 f_\bt)(u)
\label{eq:Leibniz-again} 
\end{equation}
where $P_\bt(u)$ is a polynomial of degree at most~$|\bt|$, and 
$f_\bt \in \sS_2$, for $\bt \leq \al$. From~\eqref{eq:Weyl-seminorms},
we need only show that $T \4 f$ is polynomially bounded.

Any $T \in \sS_2'$ can be written~\cite{Schwartz66} as
$T = \delhat^\ga Q$, with $\ga \in \bN^{2N}$, where $Q$ is a
polynomially bounded continuous function on~$\bR^{2N}$.
Recalling~\eqref{eq:Leibniz-rules},
$\delhat_j Q \4 f = \delhat_j(Q \4 f) - iQ \4 \mu_j f$, and by
iteration
\begin{equation}
T \4 f = \delhat^\ga Q \4 f
= \sum_{\eps\leq\ga} \delhat^\eps(Q \4 g_\eps)
\label{eq:convol-expand} 
\end{equation}
for certain $g_\eps \in \sS_2$. Combining this
with~\eqref{eq:Leibniz-again}, we need only show that $Q \4 f$ is
polynomially bounded.

If $|Q(t)| \leq C(1 + t^2)^k$, then
\begin{align*}
|(Q \4 f)(u)|
&= |\int Q(t) e^{iu'Jt} f(u - t) \,dt| 
\leq \int C(1 + t^2)^k |f(u - t)| \,dt
\\
&\leq \int 2^k C(1 + u^2)^k (1 + (u - t)^2)^k |f(u - t)| \,dt
= K(1 + u^2)^k,
\end{align*}
where $K = 2^k C \int (1 + s^2)^k |f(s)| \,ds$ is finite since
$f \in \sS_2$.

Since $K$ depends continuously on $f$, the map $f \mapsto Q \4 f$ is 
continuous from $\sS_2$ into $\sV_{k+1}$. Since each $g_\eps$ 
in~\eqref{eq:convol-expand} depends continuously of $f$, and since 
$p_{k,m}(\del_j f) \leq p_{k-1,m+1}(f)$, so that 
$\del_j \: \sV_{k-1} \to \sV_k$ is continuous, we conclude that 
$f \mapsto T \4 f$ is continuous from $\sS_2$ into $\sV_{k+|\ga|+1}$ 
and hence from $\sS_2$ into~$\sO_T$.

Now fix $f \in \sS_2$ and let $T$ vary in $\sS_2'$. Then $K$ is a
multiple of $C$, so if $V$ is a neighbourhood of zero in~$\sO_T$, then
$V \cap \sV_{k+|\ga|+1}$ is a zero-neighbourhood in $\sV_{k+|\ga|+1}$
and thus contains all $T \4 f$ with $T = \delhat^\ga Q$ and
$|Q(t)| \leq C(1 + t^2)^k$ for $C \leq c_{k\ga}$ with $c_{k\ga}$ small
enough. Let
$$
B := \biggl\{ h \in \sS_2 : \int (1 + u^2)^r \,|\del^\al h(u)| \,du 
\leq \frac{1}{c_{r\al}} \text{ for all } r \in \bN,\ \al \in \bN^{2N}
\biggr\}.
$$
Then $B$ is bounded in $\sS_2$ and its polar $B^\circ$ is a 
neighbourhood of~$0$ in $\sS_2'$ such that $T \4 f \in V$ whenever
$T \in B^\circ$.
\end{proof}

\begin{remk} 
The fact that $T \x f \in \sO_M$ has been noted in~\cite{Anderson72}.
\end{remk}

A \textit{normal space of distributions} (on $\bR^{2N}$)
\cite{Horvath66} is a locally convex space $\sR$ where
$\sD \subset \sR \subset \sD'$ with continuous inclusions and $\sD$ is
dense in~$\sR$. (Here $\sD$ is the space of test functions of compact
support on~$\bR^{2N}$.)

\begin{lema} 
\label{lm:one-normal-space}
$\sV_k^m$ is a normal space of distributions.
\end{lema}

\begin{proof}
We adapt the analogous proof of Horváth~\cite{Horvath66} for
$\sS_{-k}^m$. Take $g \in \sD$ with $g(u) = 1$ for $u^2 \leq 1$ and
$0 \leq g(u) \leq 1$ for all$ u \in \bR^{2N}$. Set
$g_\eps(u) := g(\eps u)$ for $\eps > 0$. Then for $f \in \sV_k^m$ we
have $fg_\eps \in \sD^m$ (the space of $C^m$ functions of compact
support) and from~\eqref{eq:Weyl-seminorms} we get
\begin{align*}
p_{k,m}(f - fg_\eps) 
&= \sup \biggl\{ (1 + u^2)^{-k-|\al|/2} \sum_{\bt\leq\al}
\binom{\al}{\bt} \,|\del^{\al-\bt}(1 - g(\eps u)) \,\del^\bt f(u)|
: |\al| \leq m, \ u \in \bR^{2N} \biggr\} 
\\
&\leq C \sup \biggl\{ (1 + u^2)^{-k-|\al|/2} \sum_{\bt\leq\al}
\binom{\al}{\bt} \,|\del^\bt f(u)|
: |\al| \leq m, \ u^2 \geq \eps^{-2} \biggr\}
\\
&\leq 2^m C \sup\set{(1 + u^2)^{-k-|\bt|/2} \,|\del^\bt f(u)|
: |\bt| \leq m, \ u^2 \geq \eps^{-2}},
\end{align*}
where we may take
$$
C = 1 + \sup\set{|\del^\ga g(u)| : |\ga| \leq m,\ u \in \bR^{2N}}.
$$
Thus $fg_\eps \to f$ in $\sV_k^m$ as $\eps \to 0$, and so $\sD$ is
dense in $\sV_k^m$. Hence $\sD$ is dense in $\sV_k^m$ since it is
dense in $\sD^m$. On the other hand, since
$q_{k+m,m}(f) \leq p_{k,m}(f) \leq q_{k,m}(f)$ for
$f \in \bC^m(\bR^{2N})$, we get a chain of continuous inclusions:
$$
\sD \subset \sS_{-k}^m \subset \sV_k^m \subset \sS_{-k-m}^m 
\subset \sD'.
\eqno \qed
$$
\hideqed
\end{proof}

\begin{lema} 
\label{lm:many-normal-spaces}
Let $(\sR_k)_{k\in\bN}$ be a sequence of normal spaces of
distributions. 
\begin{enumerate}
\item[\textup{(1)}]
If $\sR_{k+1} \subset \sR_k$ with a continuous inclusion for all~$k$,
and if $\sR := \bigcap_{k\in\bN} \sR_k$ with the projective topology;
\item[\textup{(2)}]
or if $\sR_k \subset \sR_{k+1}$ with a continuous inclusion for
all~$k$, and if $\sR := \bigcup_{k\in\bN} \sR_k$ with the inductive
topology;
\end{enumerate} 
then $\sR$ is a normal space of distributions.
\end{lema}

\begin{proof}
(1)~We have $\sD \subset \sR \subset \sR_k \subset \sD'$ for all 
$k \in \bN$. The first inclusion is continuous since $\sR$ has the
projective topology and each $\sD \subset \sR_k$ is continuous; the
continuity of the other inclusions is clear. If $V$ is a neighbourhood
of~$0$ in $\sR$ and if $f \in \sR$, then $V = V_k \cap \sR$ where $V_k$
is a $0$-neighbourhood in some $\sR_k$; then $f + V_k$ contains some
$g \in \sD$, and hence $g \in (f + V)$: so $\sD$ is dense in~$\sR$.

(2)~We have $\sD \subset \sR_k \subset \sR \subset \sD'$ for all
$k \in \bN$. The third inclusion is continuous since $\sR$ has the
inductive topology and each $\sR_k \subset \sD'$ is continuous; the
continuity of the other inclusions is clear. If $V$ is a neighbourhood
of~$0$ in $\sR$ and if $f \in \sR$, then $f \in \sR_k$ for some $k$ and
$V \cap \sR_k$ is a $0$-neighbourhood in some $\sR_k$; then
$f + (V \cap \sR_k)$ contains some $g \in \sD$, and hence
$g \in (f - V)$: so $\sD$ is dense in~$\sR$.
\end{proof}

Already in~\cite{Horvath66}, $\sO_C$ has been shown to be a normal
space of distributions, where the proof technique is essentially the
application of Lemma~\ref{lm:many-normal-spaces} to the definition
\eqref{eq:OC-filtration} of~$\sO_C$. From~\eqref{eq:OC-filtration} we
also obtain the normality of $\sO_M$. Combining the two Lemmas with
the definition~\eqref{eq:OT-filtration} of $\sO_T$, we get normality
of~$\sO_T$. Thus each inclusion in the chain
$$
\sD \subset \sO_C \subset \sO_T \subset \sO_M \subset \sD'
$$
is continuous and has dense image (since $\sD$ is dense in all these
spaces). Therefore the transposed maps
$$ 
\sD \subset \sO_M' \subset \sO'_T \subset \sO_C' \subset \sD'
$$
are one-to-one and continuous. We identify each dual space with its image in 
$\sD'$. Since we can interpolate $\sS_2$ and $\sS_2'$ into both chains, the 
dual spaces consist of tempered distributions.	

The space of distributions of compact support on $\bR^{2N}$
\cite{Schwartz66} is the dual space $\sE_2'$ of~$\sE_2$. We can now
show that it is contained in the Moyal algebra.

\begin{thm} 
\label{th:regular-subspaces}
The spaces $\sE_2'$, $\sO_M'$ and $\sO'_T$ are contained in~$\sM$.
\end{thm}

\begin{proof}
Transposing $\sO_T \subset \sO_M \subset \sE_2$, we get
$\sE_2' \subset \sO_M' \subset \sO'_T$, so we need only check that
$\sO'_T \subset \sM$.

If $S \in \sO'_T$ and $T \in \sS_2'$, we may define $S \x T$, $T \x S$
by transposition:
\begin{equation}
\duo{S \x T}{h} := \duo{S}{T \x h}, \qquad
\duo{T \x S}{h} := \duo{S}{h \x T},
\label{eq:more-products} 
\end{equation}
for $h \in \sS_2$; since the right hand sides are continuous in $h$,
by Theorem~\ref{th:Moyal-regularity}, $S \x T$ and $T \x S$ are
defined in~$\sS_2'$.

Moreover, for a fixed $h \in \sS_2$, the maps $T \mapsto T \x h$,
$T \mapsto h \x T$ are continuous $\:\sS_2'\to\sO_T$ by 
Theorem~\ref{th:Moyal-regularity}, so they transpose to continuous
maps $S \mapsto h \x S$, $S \mapsto S \x h$ from $\sO'_T$ into
$(\sS_2')' = \sS_2$, via
\begin{equation}
\duo{h \x S}{T} := \duo{S}{T \x h}, \qquad
\duo{S \x h}{T} := \duo{S}{h \x T}.
\label{eq:more-transposes} 
\end{equation}
Combining \eqref{eq:more-products} and~\eqref{eq:more-transposes}, we
get
\begin{equation}
\duo{T \x S}{h} := \duo{T}{S \x h}, \qquad
\duo{S \x T}{h} := \duo{T}{h \x S},
\label{eq:more-extensions} 
\end{equation}
for $T \in \sS_2'$, $S \in \sO'_T$, $h \in \sS_2$. We have shown that
$S \x h \in \sS_2$, $h \x S \in \sS_2$ whenever $h \in \sS_2$,
$S \in \sO'_T$; these are resultants of separately continuous
extensions to $\sO'_T \x \sS_2$ of the twisted product on $\sS_2$, and
since $\sS_2$ is dense in~$\sO'_T$ these extensions are unique. Since
\eqref{eq:more-extensions} is formally identical
with~\eqref{eq:Moyal-bimodule}, the twisted products
\eqref{eq:more-products} and~\eqref{eq:more-transposes} are consistent
with previous ones, and we conclude that $\sO'_T \subset \sM$. Since
$\one \notin \sO'_T$, we have $\sO'_T \neq \sM$.
\end{proof}

\begin{corl} 
\label{cr:regular-subspaces}
The space $\sO_C$ is contained in $\sM$. In particular, the (ordinary)
convolution of any function in $\sS_2$ with any tempered distribution
belongs to~$\sM$.
\end{corl}

\begin{proof}
Since $\sO_C$ is reflexive~\cite{Grothendieck55}, it is enough to note
that $\sF(\sO_M') = \sO_C$. Then use Theorem~4. Also~\cite{Horvath66},
$T * f$ belongs to $\sO_C$ if $f \in \sS_2$, $T \in \sS_2'$.
\end{proof}

If $R,S \in \sS_2'$, their tensor product $R \ox S \in \sS'(\bR^{4N})$
is given by $\duo{R \ox S}{f \ox g} := \duo{R}{f}\,\duo{S}{g}$. If 
$R,S \in \sM$, we have
\begin{align*}
\duo{R \4 S}{h} &= \duo{R}{\check S \4 h} 
= \duo{R}{u \mapsto \duo{S}{\eps_{-u} \tau_{-u} h}}
\\
&= \duo{R \ox S}{(u,v) \mapsto e^{-iu'Jv} h(u + v)}.
\end{align*}
Writing $h_2(u,v) := e^{-iu'Jv} h(u + v)$, we find that
\begin{equation}
\del_u^\al \del_v^\ga h_2(u,v)
= \sum_{\bt\leq\al} \sum_{\eps\leq\ga} i^{|\eps|-|\bt|} 
\binom{\al}{\bt} \binom{\ga}{\eps} (Jv)^\bt (Ju)^\eps \,e^{-iu'Jv}
(\del_u^{\al-\bt} \del_v^{\ga-\eps} h)(u + v).
\label{eq:deriv-expansion} 
\end{equation}
Thus $h \mapsto h_2$ is continuous as a map $\: \sE_2 \to \sE(\bR^{4N})$. Since
$(R,S) \mapsto R \ox S$ is jointly continuous
$\: \sE'_2 \x \sE'_2 \to \sE'(\bR^{4N})$ and
$\duo{R \4 S}{h} = \duo{R \ox S}{h_2}$, we find that
$(R,S) \mapsto R \4 S$ is jointly continuous on~$\sE'_2$.

By the Paley--Wiener theorem, $F(\sE'_2) = \sF(\sE'_2) =: \sOexp$ is
the space of functions in $\sO_M$ which extend to analytic functions
of exponential type, and by Theorem~\ref{th:regular-subspaces} and the
Fourier-invariance of $\sM$, $\sOexp$ is contained in~$\sM$.
($\sOexp$ carries the topology induced by $\sF$ from~$\sE'_2$.)

For convenience, we write $\hat\mu_j f := \mu_{j+N}f$ if $j \leq N$,
$\hat\mu_j f := -\mu_{j-N}f$ if $j > N$. If $h \in \sE_2$, the
expansion
$$
e^{-iu'Jv} h(u + v) 
= \sum_{k=0}^\infty \frac{(-i)^k}{k!}\, (u'Jv)^k h(u + v)
$$
converges uniformly on compact subsets of $\bR^{2N}$, together with
all derivatives on account of~\eqref{eq:deriv-expansion}, and this
convergence is uniform on bounded subsets of~$\sE_2$. Thus
$$
\duo{S \4 T}{h} = \sum_{k=0}^\infty \frac{(-i)^k}{k!}\,
\duo{S \ox T}{(u,v) \mapsto (u'Jv)^k h(u + v)}.
$$
Since $(u'Jv)^k 
= \sum_{|\al|=k} \frac{k!}{\al!}\, u^{\al\prime} (Jv)^\al$, and
$(\hat\mu_j f)(v) = (Jv)_j f(v)$ for all $f \in \sS_2$ and each~$j$,
we derive
\begin{align}
\duo{S \4 T}{h}
&= \sum_{k=0}^\infty \sum_{|\al|=k} \frac{(-i)^k}{\al!}\,
\duo{S \ox T}{(u,v) \mapsto u^{\al\prime} (Jv)^\al h(u + v)}
\nonumber \\
&= \sum_{\al\in\bN^{2N}} \frac{(-i)^{|\al|}}{\al!}\,
\duo{\mu^\al S \ox \hat\mu^\al T}{(u,v) \mapsto h(u + v)}
\nonumber \\
&= \sum_{\al\in\bN^{2N}} \frac{(-i)^{|\al|}}{\al!}\,
\duo{\mu^\al S * \hat\mu^\al T}{h} 
\label{eq:convol-expansion} 
\end{align}
where the series converges uniformly for $h$ in bounded subsets of
$\sE_2$. We are now able to expand the twisted product as a series of
products of derivatives.

\begin{thm} 
\label{th:Moyal-expansion}
If $S,T \in \sOexp$, then 
$$
S \x T = \sum_{\al\in\bN^{2N}} \frac{i^{|\al|}}{\al!}
\,(\del^\al S)\,(\delhat^\al T)
$$
with convergence in the topology of~$\sOexp$.
\end{thm}

\begin{proof}
We apply the Fourier transform $\sF$ to \eqref{eq:convol-expansion},
replacing $S,T$ by $\sF^{-1}S,\sF^{-1}T$. By the continuity of~$\sF$
we get
\begin{align*}
S \x T &= \sum_\al \frac{(-i)^{|\al|}}{\al!}\,
\sF(\mu^\al \sF^{-1}S * \hat\mu^\al \sF^{-1}T)
\\
&= \sum_\al \frac{(-i)^{|\al|}}{\al!}\, i^{2|\al|}
\,(\del^\al S)(\delhat^\al T)  
= \sum_\al \frac{i^{|\al|}}{\al!} \,(\del^\al S)\,(\delhat^\al T). 
\tag*{\qed}
\end{align*}
\hideqed
\end{proof}

\begin{corl} 
\label{cr:Moyal-bracket}
If $S,T \in \sOexp$, then
\begin{equation}
S \x T - T \x S = 2i \sum_{r=0}^\infty \sum_{|\al|=2r+1} 
\frac{(-1)^r}{\al!} \,(\del^\al S)\,(\delhat^\al T).
\label{eq:Moyal-bracket} 
\end{equation}
\end{corl}

\begin{remk} 
The restriction to $\sE'_2$ or $\sOexp$ is only needed to guarantee
convergence of the series indexed by~$\al$. If either $S$ or~$T$ is a
polynomial, these series are finite sums, and the expansions are
valid.
\end{remk}

\begin{remk} 
For $N = 1$, the leading term in~\eqref{eq:Moyal-bracket} is the
ordinary Poisson bracket $\del_q S \del_p T - \del_p S \del_q T$. As a
differential operator on $S \ox T$, we may formally write
\begin{align*}
S \x T - T \x S 
&= 2i \biggl( \sum_{r=0}^\infty \frac{(-1)^r}{(2r+1)!}\,
(\del_q \ox \del_p - \del_p \ox \del_q)^{2r+1} \biggr) (S \ox T)
\\ 
&= 2i \bigl( \sin(\del_q \ox \del_p - \del_p \ox \del_q) \bigr) 
(S \ox T),
\end{align*}
an expression first derived by Moyal~\cite{Moyal49} and called the
``Moyal bracket''.
\end{remk}

\section{The matricial form of the twisted product} 
\label{sec:matrix-basis}

By working in the Schwartz space $\sS_2$, we avoid the usual
continuity problems for the creation and annihilation operators for
the harmonic oscillator, as has been observed
before~\cite{KristensenMP65}. In the present context, these operators
are represented by first-degree polynomials in $\sM$. To avoid
notational clutter, we will take $N = 1$ in this section and the next;
but the results go through in the general case with the systematic use
of multi-indices. We write $u = (q,p)$ and use $q,p,\del_q,\del_p$ in
place of $\mu_1,\mu_2,\del_1,\del_2$ respectively. We introduce the
notations:
\begin{align*}
a &:= \frac{q + ip}{\sqrt{2}}, \quad
\abar := \frac{q - ip}{\sqrt{2}}\,,
\\
\pd{}{a} &:= \frac{\del_q - i\del_p}{\sqrt{2}}, \quad 
\pd{}{\abar} := \frac{\del_q + i\del_p}{\sqrt{2}}\,,
\\
H &:= a\abar = \half(q^2 + p^2) = \half u^2, \quad
f_0 = 2 e^{-\abar a} = 2 e^{-H}.
\end{align*}

{}From~\eqref{eq:Weyl-operators} we obtain the formulas
\begin{alignat}{2}
a \x f &= af + \pd{f}{\abar}\,,
&\qquad f \x a &= af - \pd{f}{\abar}\,,
\nonumber \\
\abar \x f &= \abar f - \pd{f}{a}\,,
&\qquad f \x \abar &= \abar f + \pd{f}{a}\,.
\label{eq:one-step} 
\end{alignat}

We get at once the following equalities in $\sM$:
\begin{equation}
\abar \x a = H - 1, \quad a \x \abar = H + 1; \quad
a \x \abar - \abar \x a = 2.
\label{eq:two-step} 
\end{equation}

The third equality is of course the canonical commutation relation for
$a$ and~$\abar$: recall that we have taken units in which $\hbar = 2$.
We note also that $a \x a \xyx a \text{ ($n$~times) } = a^n$.

The Gaussian function $f_0$ has several nice properties: it is
(pointwise) positive, it is a fixed point for the various Fourier
transforms, and it is an idempotent in~$\sS_2$ for both the twisted
product and the twisted convolution. Moreover, it is a unit vector in
$L^2(\bR^2)$, because of our choices of normalization.

Notice that if $g \in \sO_M$, \eqref{eq:one-step} implies that
\begin{equation}
a \x (gf_0) = \Bigl( \pd g\abar \Bigr) f_0, \quad 
 \abar \x (gf_0) = \Bigl( 2\abar g - \pd ga \Bigr) f_0.
\label{eq:basic-steps} 
\end{equation}
Taking $g = 2^m \abar^m$, we find that
$$
\abar \x (2^m \abar^m f_0)
= \Bigl( 2^{m+1}\abar^{m+1} - 2^m \pd{\abar^m}{a} \Bigr) f_0
= 2^{m+1} \abar^{m+1} f_0,
$$
so we get by induction that $\abar^m \x f_0 = 2^m \abar^m f_0$ if 
$m \in \bN$. If $n > m$, this gives
$$
a^n \x \abar^m \x f_0 = a^n \x (2^m \abar^m f_0) = 
2^m \,\frac{\del^n}{\del\abar^n} (\abar^m) f_0 = 0,
$$
and if $n < m$, then
$$
f_0 \x a^n \x \abar^m = (a^m \x \abar^n \x f_0)^* = 0.
$$
Also,
\begin{align*}
f_0 \x a^n \x \abar^n \x f_0
&= f_0 \x a^n \x (2^n \abar^n f_0)
= f_0 \x 2^n\, \frac{\del^n}{\del\abar^n} (\abar^n) f_0
\\
&= 2^n n!\, f_0 \x f_0 = 2^n n!\, f_0.
\end{align*}
To summarize:
\begin{equation}
f_0 \x a^n \x \abar^m \x f_0 = \dl_{mn} 2^n n!\, f_0 
\word{for} m,n \in \bN.
\label{eq:vev} 
\end{equation}

We now introduce an orthonormal basis for $L^2(\bR^2)$, which we
declare as a doubly indexed family of functions in $\sS_2$, since this
family forms a system of matrix units with respect to the twisted
product.

\begin{defn} 
\label{df:matrix-basis}
For $m,n \in \bN$, we define $f_{mn} \in \sS_2$ by
\begin{equation}
f_{mn} := \frac{1}{\sqrt{2^{m+n}\,m!\,n!}}\, \abar^m \x f_0 \x a^n.
\label{eq:matrix-basis} 
\end{equation}
\end{defn}

{}From~\eqref{eq:vev} we get directly
\begin{align}
f_{mn} \x f_{kl}
&= (2^{m+n+k+l}\,m!\,n!\,k!\,l!)^{-1/2}
\abar^m \x f_0 \x a^n \x \abar^k \x f_0 \x a^l
\nonumber \\
&= \frac{\dl_{nk}}{\sqrt{2^{m+l}\,m!\,l!}}\, \abar^m \x f_0 \x a^l 
= \dl_{nk} f_{ml}.
\label{eq:matrix-units} 
\end{align}
This implies that
\begin{align*}
2\scal{f_{mn}}{f_{kl}}
&= \duo{f_{nm}}{f_{kl}} = \int (f_{nm} \x f_{kl})(u) \,du 
= \dl_{mk} \int f_{nl}(u) \,du
\\
&= \frac{\dl_{mk}}{\sqrt{2^{n+l}\,n!\,l!}} 
\int (\abar^n \x f_0 \x f_0 \x a^l)(u) \,du
\\
&= \frac{\dl_{mk}}{\sqrt{2^{n+l}\,n!\,l!}}
\int (f_0 \x a^l \x \abar^n \x f_0)(u) \,du
\\
&= \dl_{mk} \,\dl_{nl} \int f_0(u) \,du = 2 \,\dl_{mk} \,\dl_{nl}
\end{align*}
so that $\set{f_{mn} : m,n \in \bN}$ is orthonormal in $L^2(\bR^2)$.
This family is complete since -- see \eqref{eq:basis-change} below --
all the Hermite functions on~$\bR^2$ are linear combinations of
the~$f_{mn}$.

\begin{remk} 
The basis $(f_{mn})$ lies in $\sS_2$, as is clear
from~\eqref{eq:matrix-basis} since $a,\abar \in \sM$. It has also the
important property of diagonalizing the Fourier transforms; one
readily checks that
\begin{equation}
\sF(f_{mn}) = (-i)^{m+n} f_{mn}, \quad
F(f_{mn}) = (-1)^n f_{mn}, \quad
\Fbar(f_{mn}) = (-1)^m f_{mn},
\label{eq:Fourier-eigenbasis} 
\end{equation}
and hence also $\check f_{mn} = (-1)^{m+n} f_{mn}$.
\end{remk}

To present $f_{mn}$ explicitly, we use polar coordinates
$q + ip =: \rho e^{i\al}$; note that $\rho^2 = q^2 + p^2 = u^2$.

By induction, applying \eqref{eq:basic-steps}
to~\eqref{eq:matrix-basis}, we derive
$$
f_{mn} = \frac{1}{\sqrt{2^{m+n}\,m!\,n!}} \sum_{k=0}^n (-1)^k
\binom{m}{k} \binom{n}{k}\, k!\, 2^{m+n-k} \abar^{m-k} a^{n-k} f_0
$$
and, since $a = (1/\sqrt{2})\,\rho e^{i\al}$, 
$\abar = (1/\sqrt{2})\,\rho e^{-i\al}$, we get, after some
rearrangement,
\begin{equation}
f_{mn}(\rho,\al) 
= 2(-1)^n \sqrt{\frac{n!}{m!}}\, e^{-i\al(m-n)} \rho^{m-n}
L_n^{m-n}(\rho^2) \,e^{-\rho^2/2}
\label{eq:Laguerre-basis} 
\end{equation}
and in particular
\begin{equation}
f_{nn}(\rho,\al) = 2(-1)^n L_n(\rho^2) \,e^{-\rho^2/2}
\label{eq:diagonal-basis} 
\end{equation}
where $L_n$ and $L_n^{m-n}$ are the usual Laguerre polynomials of 
order~$n$.

\begin{remk} 
Equation~\eqref{eq:diagonal-basis} agrees with the Wigner function of
the $n$th energy level of the harmonic oscillator, obtained
in~\cite{Groenewold46}; and \eqref{eq:Laguerre-basis} agrees with the
``transition'' between levels of the harmonic oscillator, first
derived in~\cite{BartlettM49}; see also \cite{AmietH81}
and~\cite{Howe80}.
\end{remk}

Now we represent $\sS_2$, $\sS_2'$ and $L^2(\bR^2)$ as sequence spaces
of coefficients after expansion in the twisted Hermite basis
$(f_{mn})$. Our treatment is in the spirit of Simon's
work~\cite{Simon71} with the ordinary Hermite basis. Since here the
coefficients form a doubly indexed family, we may consider their
matrix product, which turns out to correspond to the twisted product
of the associated functions or distributions.

The fundamental fact which underlies the sequence constructions is
that the twisted Hermite basis ``diagonalizes'' the oscillator
Hamiltonian $H$ (its eigenvalues are odd integers rather than
half-integers due to our convention that $\hbar = 2$).

\begin{prop} 
\label{pr:energy-levels}
If $m,n \in \bN$, then
$$
H \x f_{mn} = (2m + 1) f_{mn}, \qquad f_{mn} \x H = (2n + 1) f_{mn}.
$$
\end{prop}

\begin{proof}
{}From \eqref{eq:matrix-basis} we get at once
\begin{alignat*}{2}
a \x f_{mn} &= \sqrt{2m}\, f_{m-1,n}, 
&\qquad f_{mn} \x a &= \sqrt{2n + 2}\, f_{m,n+1}
\\
\abar \x f_{mn} &= \sqrt{2m + 2}\, f_{m+1,n},
&\qquad f_{mn} \x \abar &= \sqrt{2n}\, f_{m,n-1},
\end{alignat*}
(with $f_{mn} = 0$ if $m$ or $n$ is~$-1$). The result follows
from~\eqref{eq:two-step}.
\end{proof}

Let us write $A(f) := H \x f \x H$. We could consider $A$ as an
operator on $L^2(\bR^2)$ with domain $\sS_2$; as such, $A$ is
symmetric and closable, and clearly unbounded. Indeed,
$(A \pm iI)f_{mn} = ((2m+1)(2n+1) \pm i)f_{mn}$ and hence $A \pm iI$
has dense range: thus $A$ is essentially self-adjoint. Moreover,
$A^{-1}$ has finite-dimensional eigenspaces and
$$
\sum_{m,n=0}^\infty (2m+1)^{-2} (2n+1)^{-2} = (\pi^2/8)^2
$$
is finite, so $A^{-1}$ is a Hilbert--Schmidt operator. It is not hard
to check that the seminorms $f \mapsto \|A^k f\|$, for $k \in \bN$,
generate the topology of~$\sS_2$.

\begin{remk} 
Let $B = u^2 - \Dl$ be the usual Hermite operator on $L^2(\bR^2)$. We 
have $Bf = H \x f + f \x H$, so $Bf_{mn} = 2(m+n+1)f_{mn}$. If
$h_k(x) := (2^{k-1}\,k!)^{-1/2} H_k(x) e^{-x^2/2}$ is the usual
Hermite function of degree~$k$, we conclude that
\begin{equation}
f_{mn} = \sum_{k+l=m+n} c_{mn}^{kl} h_k \ox h_l,  \qquad 
h_k \ox h_l = \sum_{m+n=k+l} b_{kl}^{mn} f_{mn},
\label{eq:basis-change} 
\end{equation}
for some constants $c_{mn}^{kl}$, $b_{kl}^{mn}$. In fact, we may
compute that
$$
c_{mn}^{kl} = 2^{(m-n)/2} i^{2m+l} {\binom{m+n}{l}}^{1/2}
{\binom{m+n}{m}}^{-1/2} P_m^{l-m,k-m}(0)
$$
where $P_m^{l-m,k-m}$ is the usual Jacobi polynomial. (This takes care
of the completeness argument for the~$f_{mn}$.)
\end{remk}

We can now characterize $\sS_2$ and~$\sS_2'$ as sequence spaces.

\begin{thm} 
\label{th:sequence-space}
Let $\sss$ be the Fréchet space of rapidly decreasing double
sequences $c = (c_{mn})$ such that
$$
r_k(c) := \biggl[
\sum_{m,n=0}^\infty (2m+1)^{2k} (2n+1)^{2k} |c_{mn}|^2 \biggr]^{1/2}
$$
is finite for all $k \in \bN$, topologized by the seminorms
$(r_k)_{k\in\bN}$. For $f \in \sS_2$, let $c$ be the sequence of
coefficients in the expansion $f = \sum_{m,n=0}^\infty c_{mn} f_{mn}$.
Then $f \mapsto c$ is an isomorphism of Fréchet spaces from $\sS_2$
onto~$\sss$.
\end{thm}

\begin{proof}
If $f \in \sS_2$, then $\|A^k f\| < \infty$ for all $k \in \bN$, so 
that $r_k(c) = \|A^k f\|$ is finite for all~$k$. It follows that
$f \mapsto c$ is a one-to-one topological isomorphism of $\sS_2$
into~$\sss$.

Given any $c \in \sss$, for $M,N \in \bN$ let $c^{MN}$ be the double
sequence defined by $c_{mn}^{MN} := c_{mn}$ if $m \leq M$, $n \leq N$
and $c_{mn}^{MN} := 0$ otherwise. Then $r_k(c^{MN} - c) \to 0$ as
$M,N\to\infty$, for each $k$, so that the functions
$\sum_{m=0}^M \sum_{n=0}^N c_{mn} f_{mn}$ form a Cauchy sequence
in~$\sS_2$ and hence converge to a function $f$ which maps onto~$c$.
\end{proof}

For $T \in \sS_2'$, $m,n \in \bN$, define $b_{mn} := \duo{T}{f_{nm}}$.
Then
\begin{equation}
\duo{T,f} = \sum_{m,n=0}^\infty c_{nm} \duo{T}{f_{nm}}
= \sum_{m,n=0}^\infty c_{nm} b_{mn}
\label{eq:matrix-expansion} 
\end{equation}
where the series converges absolutely, for each
$f = \sum_{m,n=0}^\infty c_{mn}f_{mn} \in \sS_2$. Since
$T \in \sS_2'$, there exist $k \in \bN$, $K > 0$ such that
\begin{equation}
|\duo{T,f}| \leq K \|A^k f\| = K\,r_k(c)
\label{eq:sequence-estimate} 
\end{equation}
for all $f \in \sS_2$, and since
$$
\duo{T,f} = \sum_{m,n=0}^\infty (2m + 1)^{-k} (2n + 1)^{-k} b_{mn}
(2n + 1)^k (2m + 1)^k c_{nm}
$$
the Schwarz inequality gives $r_{-k}(b) \leq K$. Thus, whenever
$b_{mn}$ is a double sequence with $r_{-k}(b)$ finite for some~$k$,
the series $\sum_{m,n=0}^\infty b_{mn}f_{mn}$ converges weakly to~$T$
in $\sS_2'$, and \eqref{eq:sequence-estimate} shows that the
convergence is uniform on bounded subsets of~$\sS_2$, so the series
converges to $T$ in the strong dual topology of~$\sS_2'$.

The main result of this section is now easy.

\begin{thm} 
\label{th:matrix-basis}
If $a,b \in \sss$ correspond respectively to $f,g \in \sS_2$ as
coefficient sequences in the twisted Hermite basis, then the sequence
corresponding to the twisted product $f \x g$ is the matrix
product~$ab$, where
\begin{equation}
(ab)_{mn} := \sum_{k=0}^\infty a_{mk} b_{kn}.
\label{eq:matrix-product} 
\end{equation}
\end{thm}

\begin{proof}
{}From \eqref{eq:matrix-units} and the continuity of~$\x$ in $\sS_2$,
we get
\begin{align*}
f \x g
&= \biggl(\sum_{m,k} a_{mk} f_{mk} \biggr)
\x \biggl(\sum_{r,n} b_{rn} f_{rn} \biggr) 
\\  
&= \sum_{m,k,r,n} a_{mk} b_{rn}\, f_{mk} \x f_{rn} 
= \sum_{m,k,n} a_{mk} b_{kn}\, f_{mn}.
\tag*{\qed}
\end{align*}
\hideqed
\end{proof}

\begin{corl} 
\label{cr:Schwartz-square}
$\set{f \x g : f,g \in \sS_2} = \sS_2$.
\end{corl}

\begin{proof}
It suffices to show that any $c \in \sss$ can be written as the matrix
product of two sequences in~$\sss$. We use Howe's
argument~\cite{Howe80} to show this.

Set 
$d_m := \bigl( \sup\set{|c_{jr}| : j \in \bN,\ r \geq m} \bigr)^{1/2}$
for $m \in \bN$, and let $d$ be the ``diagonal'' sequence with entries
$d_m\dl_{mn}$. Then one verifies that $r_k(d)^2 \leq C_k r_{2k+2}(c)$ 
for some constants $C_k$, so that $d \in \sss$. Now if we set
$b_{mn} := c_{mn}/d_n$, we get 
$|b_{mn}| = |c_{mn}|/d_n \leq d_n^2/d_n = d_n$ and thus $b \in \sss$
also. Clearly $bd = c$.
\end{proof}

\begin{remk} 
The sequence $a \4 b$ corresponding to $f \4 g$ is
$$
(a \4 b)_{mn} := \sum_{k=0}^\infty (-1)^k a_{mk} b_{kn}.
$$
Since, by \eqref{eq:Fourier-switch} and~\eqref{eq:Fourier-cross},
$f \4 g = f \x \dl \x g$, it suffices to show that $\dl$ is
represented by the diagonal matrix with entries $(-1)^m \dl_{mn}$;
this follows from~\eqref{eq:Fourier-eigenbasis}, since
$\sF\one = \dl$. Thus the entire theory of the twisted product and
convolution could be developed, at least formally, in the matrix
language and without mention of the symplectic Fourier transforms;
the basic transformation formulas -- see~\eqref{eq:Fourier-switch} --
are:
$$
R \4 S = R \x \dl \x S, \qquad R \x S = R \4 \one \4 S.
$$
\end{remk}

We now show that \eqref{eq:matrix-product} gives a second way of
defining the twisted product for many pairs of distributions, which
lie in spaces of Sobelev type.

\begin{defn} 
\label{df:double-scale}
For $s,t \in \bR$, we denote by $\sG_{s,t}$ the Hilbert space obtained
by completing $\sS_2$ with respect to the norm
\begin{equation}
\|f\|_{s,t} := \biggl( \sum_{m,n=0}^\infty (2m + 1)^s (2n + 1)^t
|c_{mn}|^2 \biggr)^{1/2}.
\label{eq:st-norms} 
\end{equation}
\end{defn}

Observe that $\sG_{0,0} = L^2(\bR^2)$ with $\|f\|_{0,0} = \|f\|$. An
orthonormal basis for $\sG_{s,t}$ is given by the functions
$(2m + 1)^{-s/2} (2n + 1)^{-t/2} f_{mn}$, and thus
$$
f = \sum_{m,n=0}^\infty c_{mn} f_{mn}
$$
with convergence in the $(s,t)$-norm, for all $f \in \sG_{s,t}$. Note
that $\sS_2 = \bigcap_{s,t\in\bR} \sG_{s,t}$ topologically. Since
$\sS_2 \subset \sG_{s,t}$ is a continuous inclusion with dense image,
$\sG_{s,t}$ is a normal space of tempered distributions, and the
transpose of $\sS_2 \subset \sG_{s,t}$ is the inclusion
$\sG_{-t,-s} \subset \sS_2'$. Also, from~\eqref{eq:sequence-estimate},
$\sS_2' = \bigcup_{s,t\in\bR} \sG_{s,t}$ (topologically). Furthermore,
$\sG_{s,t} \subset \sG_{q,r}$ with a continuous inclusion iff
$s \geq q$ and~$t \geq r$. Note also that $f^* \in \sG_{t,s}$ whenever
$f \in \sG_{s,t}$.

If $g = \sum_{m,n=0}^\infty b_{mn} f_{mn} \in \sG_{q,r}$, we define
formally
\begin{equation}
f \x g := \sum_{m,n=0}^\infty \biggl(
\sum_{k=0}^\infty c_{mk} b_{kn} \biggr) f_{mn}.
\label{eq:matricial-product} 
\end{equation}

\begin{thm} 
\label{th:matricial-product}
\begin{enumerate}
\item[(1)]
The serie \eqref{eq:matricial-product} converges in $\sG_{s,r}$ if
$t + q \geq 0$, and in that case
$$
\|f \x g\|_{s,r} \leq \|f\|_{s,t}\, \|g\|_{q,r}\,.
$$
\item[(2)]
$\sG_{s,t}$ is a Banach algebra under the twisted product 
\eqref{eq:matricial-product} whenever $s + t \geq 0$; for $s \geq 0$,
$\sG_{s,s}$~is a Banach $*$-algebra.
\item[(3)]
The Fourier transforms $F,\Fbar,\sF$ are unitary isometries of each 
$\sG_{s,t}$ onto itself.
\item[(4)]
The twisted product \eqref{eq:matricial-product} is consistent with
previous definitions.
\end{enumerate}
\end{thm}

\begin{proof}
(1) From the Schwarz inequality we get
\begin{align*}
\|f \x g\|_{s,r}^2 
&\leq \sum_{m,n=0}^\infty (2m + 1)^s \biggl(
\sum_{k=0}^\infty |c_{mk} b_{kn}| \biggr)^2 (2n + 1)^r
\\
&= \sum_{m,n=0}^\infty (2m + 1)^s \biggl( \sum_{k=0}^\infty |c_{mk}|
(2k + 1)^{t/2} (2k + 1)^{-t/2} |b_{kn}| \biggr)^2 (2n + 1)^r
\\
&\leq \sum_{m,k=0}^\infty (2m + 1)^s |c_{mk}|^2 (2k + 1)^t 
\sum_{l,n=0}^\infty (2l + 1)^{-t} |b_{ln}|^2 (2n + 1)^r
\\
&= \|f\|_{s,t}^2 \,\|g\|_{-t,r}^2 \leq \|f\|_{s,t}^2 \,\|g\|_{q,r}^2
\end{align*}
whenever $q \geq -t$, which yields convergence
of~\eqref{eq:matricial-product} in this case.

(2) follows by taking $q = s$, $r = t$.

(3) The Fourier invariance of $\sG_{s,t}$ is evident from
\eqref{eq:st-norms} and~\eqref{eq:Fourier-eigenbasis}.

(4) If $T = \sum_{m,n=0}^\infty d_{mn} f_{mn} \in \sS_2'$, then
\begin{align}
\duo{T}{f \x g}
&= \sum_{m,n=0}^\infty d_{mn} \biggl(
\sum_{k=0}^\infty c_{mk} b_{kn} \biggr)
\nonumber \\
&= \sum_{k,n=0}^\infty \biggl( 
\sum_{m=0}^\infty d_{nm} c_{mk} \biggr) b_{kn} = \duo{T \x f}{g},
\label{eq:matricial-transfer} 
\end{align}
where the convergence of the double sums is absolute
by~\eqref{eq:matrix-expansion} and that of the simple sums is also
absolute, in the second case because $T \x f$ lies in some $\sG_{s,t}$
and $g$ lies in $\sS_2 \subset \sG_{-t,-s}$; thus we may interchange
the summations, obtaining
$$
T \x f = \sum_{k,n=0}^\infty \bigl(
\sum_{m=0}^\infty d_{nm} c_{mk} \bigr) f_{nk}.
$$

If $f \in \sM_L$, $g \in \sS_2$, then by
Definition~\ref{df:Moyal-bimodule},
$$
\duo{T \x f}{g} := \duo{T}{f \x g} 
= \sum_{m,n=0}^\infty d_{nm} \biggl(
\sum_{k=0}^\infty c_{mk} b_{kn} \biggr);
$$
the double sum converges absolutely since $f \x g \in \sS_2$, so we
may interchange the order of summation to
recover~\eqref{eq:matricial-transfer}.
\end{proof}

\begin{remk} 
We see that if $f,g \in L^2(\bR^2)$, then $f \x g \in L^2(\bR^2)$ and
$\|f \x g\| \leq \|f\|\,\|g\|$. Moreover, $f \x g$ lies in
$C_0(\bR^2)$: the continuity follows by adapting the analogous
argument for (ordinary) convolution.
\end{remk}

\begin{remk} 
Notice from \eqref{eq:st-norms} that $A(\sG_{s,t}) = \sG_{s-2,t-2}$,
where we write $A(T) := H \x T \x H$ for any $T \in \sS_2'$, which
makes sense since $H \in \sM$. Clearly $A(\sM) \subset \sM$, so that
if $\sM$ were to contain any $\sG_{s,t}$, it would contain them all.
However, $\sM \neq \sS_2'$, and thus $\sM$ contains no $\sG_{s,t}$; in
particular, $L^2(\bR^{2N}) \nsubseteq \sM$. Hence $\sM$ really
provides a \textit{different} extension of~$\x$ on $\sS_2$ from the
Banach algebras $\sG_{s,t}$ $(s + t \geq 0)$.
\end{remk}

\section{Conclusion and outlook} 
\label{sec:outlook}

We have been led to define the twisted product of a pair of
distributions and to introduce the Moyal $*$-algebra $\sM$ under the
twisted product. A rigorous formulation of the phase-space approach to
quantum theory, in the arena given by $\sM$, should address the
following problems:
\begin{enumerate}
\item[(a)]
find an equivalent of the spectral theorem;
\item[(b)]
solve the dynamical equations using such a spectral theorem;
\item[(c)]
describe the algebraic structure of the state space corresponding
to~$\sM$.
\end{enumerate}

We plan to take up these problems. Our first task is to give $\sM$ an
appropriate topology: this we do in the following paper~\cite{Deimos}.

\subsection*{Acknowledgements}

We are grateful for helpful correspondence from John Horváth and Peter
Wagner. We would like to thank Prof.~Abdus Salam and the International
Centre for Theoretical Physics, for their hospitality during a stay in
which this work was completed. We gratefully acknowledge support from
the Vicerrectoría de Investigación of the Universidad de Costa Rica.

\end{document}